\documentclass[10pt]{amsart} 
\title[Combinatorics of the ASEP on a semi-infinite lattice]{Combinatorics of the 
asymmetric exclusion process on a semi-infinite lattice}
\author{Tomohiro Sasamoto}
\address{Department of Mathematics, Chiba university
          Yayoi-cho 1-33, Inage, Chiba 263-8522, Japan}
\email{sasamoto@math.s.chiba-u.ac.jp}
\thanks{The first author would like to acknowledge the financial support from 
          KAKENHI (2274054)}
\author{Lauren  Williams}
\address{Department of Mathematics, Evans Hall 913, 
University of California, Berkeley}
\email{williams@math.berkeley.edu}
\thanks{The second author is partially supported by 
an NSF CAREER award and an Alfred Sloan Fellowship.}



\usepackage{pstricks, pst-node, amssymb}

\psset{unit=1pt, arrowsize=4pt, linewidth=.7pt}
\psset{linecolor=blue}
\newgray{grayish}{.90}
\newrgbcolor{embgreen}{0 .5 0}

\def\vblack(#1, #2)#3{\cnode*[linecolor=black](#1, #2){3}{#3}}
\def\vwhite(#1,#2)#3{\cnode[linecolor=black,fillcolor=white,fillstyle=solid](#1,
#2){3}{#3}}
\countdef\x=23
\countdef\y=24
\countdef\z=25
\countdef\t=26

\def\tbox(#1,#2)#3{
\x=#1 \y=#2
\multiply\x by 12
\multiply\y by 12
\z=\x \t=\y
\advance\z by 12
\advance\t by 12
\psline(\x,\y)(\x,\t)(\z,\t)(\z,\y)(\x,\y)
\advance\x by 6
\advance\y by 6
\rput(\x,\y){{\bf #3}}}

\usepackage{amsmath, amsthm, amssymb, amsbsy}
\usepackage{amsfonts, latexsym, stmaryrd, amscd, xy}
\usepackage[mathscr]{eucal}
\usepackage{epsfig}

\newtheorem{theorem}{Theorem}[section]

\newtheorem{proposition}[theorem]{Proposition}

\newtheorem{example}[theorem]{Example}
\newtheorem{corollary}[theorem]{Corollary}

\newtheorem{remark}[theorem]{Remark}

\newtheorem{definition}[theorem]{Definition}

\usepackage{amsfonts}
\usepackage{xy}
\usepackage{amssymb}
\usepackage{amsmath}

\newcommand{\T}{\mathcal{T}}

\newcommand{\Z}{\mathbb Z}
\newcommand{\R}{\mathbb R}
\newcommand{\E}{\mathbb E}

\newcommand{\LL}{\mathcal{L}}

\DeclareMathOperator{\wt}{wt}

\newcommand{\thmrefer}[1]{\renewcommand\thetheorem
  {\protect\ref{#1}}\addtocounter{theorem}{-1}}

\xyoption{all}
\CompileMatrices

\begin{document}



\begin{abstract}
We study two versions of the asymmetric exclusion process (ASEP) -- 
an ASEP on a semi-infinite lattice $\Z^+$ with an open left boundary,
and an ASEP on a finite lattice with open left and right boundaries --
and we demonstrate a surprising relationship between their
stationary measures.
The semi-infinite ASEP was first studied 
by Liggett \cite{Liggett} and then
Grosskinsky \cite{G}, while the finite ASEP had been introduced 
earlier by Spitzer \cite{Spitzer} and Macdonald-Gibbs-Pipkin \cite{bio}.
We show that the 
finite correlation functions
involving the first $L$ sites
for the stationary measures on the semi-infinite ASEP can be obtained as a \emph{nonphysical} 
specialization of 
the stationary distribution of an ASEP on a finite one-dimensional 
lattice with $L$ sites.
Namely, if the output and input rates of particles at the right
boundary of the finite ASEP are $\beta$ and $\delta$, respectively,
and we set $\delta=-\beta$, then this specialization corresponds
to sending the right boundary of the lattice to infinity.
Combining this observation
with  work of the second author and Corteel \cite{CW3, CW4},
we obtain a combinatorial formula for finite correlation functions
of the ASEP on a semi-infinite lattice.
\end{abstract}

\maketitle

\section{Introduction}

The asymmetric exclusion process (ASEP) is a model in which particles hop 
on a lattice, subject to the condition that there is at most 
one particle per site. It was first introduced by 
Spitzer \cite{Spitzer} and also by Macdonal-Gibbs-Pipkin \cite{bio}
in the context of protein synthesis, who studied this model 
on a finite lattice of $L$ sites.  A version of the model 
where particles hop on the semi-infinite lattice $\Z^+$ was   
studied by Liggett \cite{Liggett}, and subsequently  
by Grosskinsky in his thesis \cite{G}.
In the semi-infinite ASEP, particles may enter and exit  at the left
boundary at rates $\alpha$ and $\gamma$, respectively, and in the bulk,
particles may hop right and left to neighboring sites
of the lattice at rates $1$ and $q$.  
Let $c$ be an additional parameter;
it winds up determining the 
\emph{stationary current} of particles.  
We denote states by  
vectors $\eta = (\eta_1, \eta_2,\dots )$, where 
$\eta_i \in \{0,1\}$, and we  
denote the set of all states by $X$.
\begin{figure}[h]
\centering
\includegraphics[height=.58in]{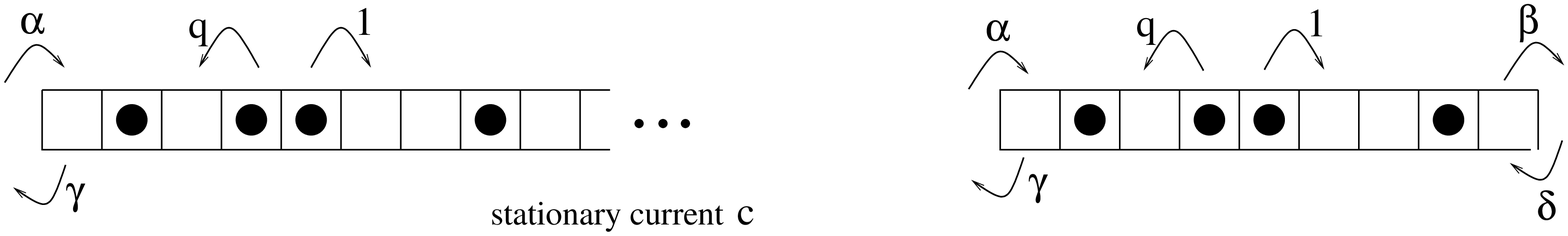}
\caption{}
\label{Models}
\end{figure}

In the ASEP on a finite one-dimensional lattice of $L$
sites, particles may enter and exit at the left boundary
at rates $\alpha$ and $\gamma$, and may exit and enter 
at the right boundary at rates $\beta$ and $\delta$.  In the bulk,
particles may hop right and left to neighboring sites of the lattice
at rates $1$ and $q$.

We refer to these two flavors of the ASEP as the semi-infinite ASEP and 
the finite ASEP.  The two models are illustrated in Figure \ref{Models}.
In both models we assume that all parameters are non-negative.

Given a measure $\mu$ on $X$,
and a word $(\eta_1,\dots,\eta_L) \in \{0,1\}^L$, the 
\emph{correlation function} 
$\langle \eta_1 \dots \eta_L \rangle$
is the expected value
with respect to $\mu$ that the leftmost $L$ sites of 
a state in the semi-infinite ASEP will be 
$\eta_1,\dots,\eta_L$.

Our first result is the following.
\begin{theorem}\label{mainthm}
The finite correlation functions
involving the leftmost $L$ sites 
of the stationary measures 
of the semi-infinite ASEP can be obtained from 
the stationary distribution for the finite ASEP
on a lattice of $L$ sites, after setting 
$\beta=c$ and $\delta=-c$.
More specifically, the correlation
function 
$\langle \eta_1 \dots \eta_L \rangle$
of the stationary measure for the semi-infinite ASEP corresponding to 
the stationary current $c$ is equal to 
$\mu^{fin}(\alpha,c,\gamma, -c;q)(\eta_1,\dots,\eta_L)$,
the quantity one obtains 
by setting $\beta=c$ and $\delta = -c$
in the steady state
probability of state $(\eta_1,\dots,\eta_L)$ in the finite 
ASEP.
\end{theorem}
This theorem is illustrated in Figure \ref{Models2}.
\begin{figure}[h]
\centering
\includegraphics[height=.56in]{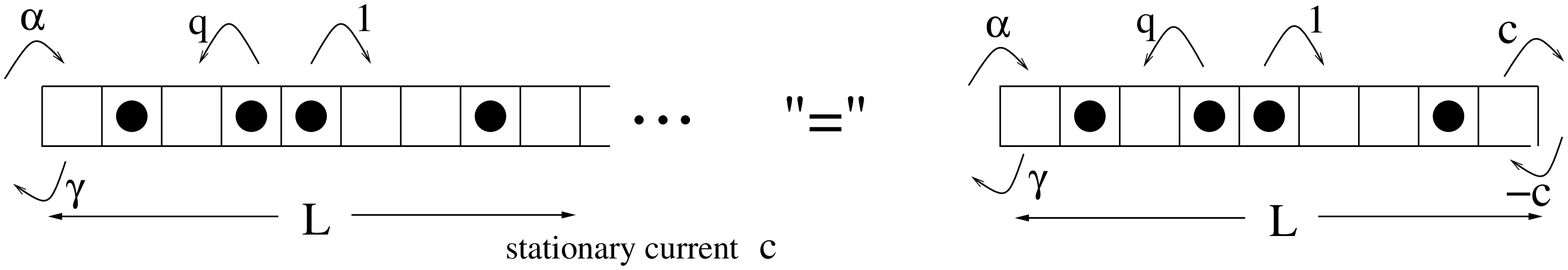}
\caption{}
\label{Models2}
\end{figure}

By combining  Theorem \ref{mainthm} with 
work of the second author and Corteel \cite{CW3, CW4}
we can  give a combinatorial formula for the 
finite correlation functions for the stationary measures
for the 
semi-infinite ASEP.  
Before stating the result,
we first introduce the {\it staircase tableaux} from
\cite{CW3, CW4}.

\begin{definition}
A \emph{staircase tableau} of size $L$ is a Young diagram of ``staircase"
shape $(L, L-1, \dots, 2, 1)$ such that boxes are either empty or
labeled with $\alpha, \beta, \gamma$, or $\delta$, subject to the following conditions:
\begin{itemize}
\item no box along the diagonal is empty;
\item all boxes in the same row and to the left of a $\beta$ or a $\delta$ are empty;
\item all boxes in the same column and above an $\alpha$ or a $\gamma$ are empty.
\end{itemize}
The \emph{type} of a staircase tableau is a word in $\{\circ, \bullet\}^L$
obtained by reading the diagonal boxes from
northeast to southwest and writing a $\bullet$ for each $\alpha$ or $\delta$,
and a $\circ$ for each $\beta$ or $\gamma$.
\end{definition}

See the left of Figure \ref{staircase} for an example of a staircase tableau.
\begin{figure}[h]
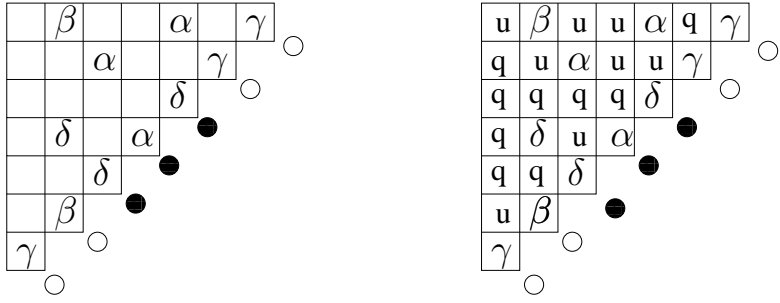

\input{Staircase.pstex_t} \hspace{6em} \input{Staircase2.pstex_t}
\caption{A staircase tableau
of size $7$ and type $\circ \circ \bullet \bullet \bullet \circ \circ$}
\label{staircase}
\end{figure}


\begin{definition}\label{weight}
The \emph{weight} $\wt(\T)$ of a staircase tableau $\T$ 
is a monomial in
$\alpha, \beta, \gamma, \delta, q$, and $u$, which we obtain as follows.
Every blank box of $\T$ is assigned a $q$ or $u$, based on the label of the closest
labeled box to its right in the same row and the label of the closest labeled box
below it in the same column, such that:
\begin{itemize}
\item every blank box which sees a $\beta$ to its right gets assigned a $u$;
\item every blank box which sees a $\delta$ to its right gets assigned a $q$;
\item every blank box which sees an $\alpha$ or $\gamma$ to its right,
  and an $\alpha$ or $\delta$ below it, gets assigned a $u$;
\item every blank box which sees an $\alpha$ or $\gamma$ to its right,
 and a $\beta$ or $\gamma$ below it, gets assigned a $q$.
\end{itemize}
After assigning a $q$ or $u$ to each blank box in this way, the
\emph{weight} of $\T$ is then defined as the product of all labels in all boxes.
\end{definition}

The right of Figure \ref{staircase}
shows that this staircase tableau
has weight $\alpha^3 \beta^2 \gamma^3 \delta^3 q^9 u^8$.

\begin{remark}
The weight of a staircase tableau of size $L$
has degree $L(L+1)/2$.  We will typically set $u=1$.
Keeping $u$ general
corresponds to particles in the bulk hopping right at rate $u$
instead of $1$.
\end{remark}

The following result, concerning
the stationary distribution of the finite
ASEP, was announced in \cite{CW3} and proved
in \cite{CW4}.

\begin{theorem}\label{OldThm}\cite{CW3, CW4}
Consider any state $\tau$ of
the finite ASEP with $L$ sites, where the parameters
$\alpha, \beta, \gamma, \delta, q, u$ are general.
Set $Z_L^{fin} = \sum_{\T} \wt(\T)$,
where the sum is over all staircase tableaux of size $L$;
we call $Z_L^{fin}$ the \emph{partition function} of the finite ASEP.
Then the steady state probability that the
ASEP is at state $\tau$ is
precisely
\begin{equation*}
\frac{\sum_{\T} \wt(\T)}{Z_L^{fin}},
\end{equation*}
where the sum is over all staircase tableaux $\T$ of type $\tau$.
\end{theorem}

\begin{example}
Figure \ref{tableaux} illustrates Theorem \ref{OldThm}
for the state $\bullet \bullet$
of the ASEP.  All staircase tableaux $\T$
of type $\bullet \bullet$ are shown.
It follows that the steady state probability of  $\bullet \bullet$ is
\begin{equation*}
\frac{\alpha^2 u + \delta^2 q +\alpha \delta q  +\alpha \delta u +
\alpha^2 \delta+ \alpha \beta \delta +\alpha \gamma \delta+\alpha \delta^2}
{Z_2^{fin}}.
\end{equation*}

\begin{figure}[h]
\input{Example.pstex_t}
\caption{The tableaux of type $\bullet \bullet$}
\label{tableaux}
\end{figure}
\end{example}

\begin{remark}\label{rem:Zsemi}
In \cite{CSSW}, the second author together with Corteel, Stanley, and Stanton, 
studied staircase tableaux and their generating function
$Z_L^{fin}(\alpha,\beta,\gamma,\delta;q)$ from a combinatorial point of view.
In particular, Table 1 of \cite{CSSW} lists various specializations
of $Z_L^{fin}(\alpha,\beta,\gamma,\delta;q)$.
The third row of the table
shows that $$Z_L^{fin}(\alpha, \beta, \gamma, -\beta;q)=
\prod_{j=0}^{L-1}(\alpha+q^j \gamma).$$  Note that 
despite the fact that the specialization is nonphysical (we have made the hopping rate
$\delta$ a negative real number), the resulting 
quantity is positive.  Also, the resulting quantity has no dependence
on $\beta$ and $\delta$.
\end{remark}

\begin{corollary}\label{comb}
Consider the semi-infinite ASEP, with  parameters
$\alpha, \gamma, q, c$.
Set $Z_L^{semi} = \prod_{j=0}^{L-1} (\alpha+q^j \gamma)$.
Then the 
correlation function $\langle \eta_1 \dots \eta_L \rangle$
for the stationary measure
is precisely
\begin{equation}\label{eq:corr}
\frac{\sum_{\T} \wt(\T)|_{u=1,\beta=c,\delta=-c}}{Z_L^{semi}},
\end{equation}
where the sum is over all staircase tableaux $\T$ of type $(\eta_1,\dots,\eta_L)$.
\end{corollary}

\begin{proof}
Corollary \ref{comb} follows from Theorem \ref{mainthm}, Theorem \ref{OldThm}, and 
Remark \ref{rem:Zsemi}.
\end{proof}


\begin{remark}
For the finite ASEP these correlation functions could be written
as polynomials in the parameters with all coefficients being positive
(divided by a normalization factor).  However, for the semi-infinite
ASEP, since we use the substition $\delta=-c$ in Corollary \ref{comb}, 
this positivity property of the coefficients no longer holds.

Nevertheless, in Theorem \ref{positivity} we will provide
a sufficient condition for the quantities 
in \eqref{eq:corr} to be positive real numbers.
\end{remark}

\begin{example}
We can use 
Corollary \ref{comb} 
and the tableaux of Figure \ref{tableaux}  
to compute the correlation function $\langle \eta_1 \eta_2 \rangle=
\langle 1 1 \rangle$.  Setting
$u=1,\beta=c,\delta=-c$ gives 
\begin{equation*}
\frac{\alpha^2  + c^2 q - \alpha c q  -\alpha c -
\alpha^2 c- \alpha c^2 - \alpha \gamma c+\alpha c^2}
{(\alpha+\gamma)(\alpha+q\gamma)} = 
\frac{\alpha^2  + c^2 q - \alpha c q  -\alpha c -
\alpha^2 c - \alpha \gamma c}
{(\alpha+\gamma)(\alpha+q\gamma)}.
\end{equation*}
\end{example}

\begin{remark}\label{partition}
The partition function  $Z_L^{semi}$ for 
the finite correlation functions involving the first $L$
sites on the semi-infinite ASEP is 
$\prod_{j=0}^{L-1}(\alpha+q^j \gamma)$.
In particular, this does not depend on $c$.
\end{remark}

The structure of this paper is as follows.  In 
Section \ref{sec:finite} we review some results on the finite ASEP,
then in Section \ref{semi-infinite} we 
define the ASEP on a semi-infinite lattice.  
In Section \ref{semiansatz} we state and prove a matrix ansatz describing
the finite correlation functions of its (signed) stationary measures,
and in Section \ref{sec:positivity} we provide a sufficient condition
for 
these signed stationary measures to be positive.
In Section \ref{proof} we 
prove Theorem \ref{mainthm}.  And in Section \ref{conclusion} we summarize our results
and end with some questions about ``nonsensical" specializations of Markov chains.

\textsc{Acknowledgments:} The authors are grateful to A. Borodin, P. Deift, 
P.L. Ferrari, S. Grosskinsky, and H. Spohn for useful discussions. 

\section{Background on the finite ASEP and its matrix ansatz} \label{sec:finite}

We start by recalling the Matrix Ansatz of Derrida, Evans, Hakim, and Pasquier
\cite{Derrida1} for the finite ASEP,  as well as
results of 
the first author together with Uchiyama and Wadati \cite{USW} on the current.

\begin{theorem} \cite{Derrida1} \label{ansatz}
Suppose that there are matrices $D$, $E$ and vectors
$\langle W |$, $|V\rangle$, which satisfy
\begin{enumerate}
\item[(A.)] $DE-qED=D+E$
\item[(B.)] $\alpha \langle W|E-\gamma \langle W|D = \langle W|$
\item[(C.)] $\beta D|V\rangle - \delta E|V\rangle = |V\rangle.$
\end{enumerate}
Let $\eta = (\eta_1,\dots,\eta_L)$ be a state of the finite ASEP.
Then the measure $\mu$ defined by 
$$\mu(\eta) = 
\frac{\langle W| \prod_{x=1}^L \eta_x D+(1-\eta_x) E|V\rangle}
{\langle W| (D+E)^L | V\rangle}$$
is the unique stationary measure for the ASEP on a finite
lattice of $L$ sites, where the rates of particles entering and exiting
at the left are $\alpha$ and $\gamma$, and the rates of particles
exiting and entering at the right are $\beta$ and $\delta$.
\end{theorem}

Although Theorem \ref{ansatz} was published in 1993, it was not until 
ten years later that 
a general solution to the ansatz was obtained.
\begin{theorem}\cite[Section 4.2]{USW}\label{USWsolution}
There is a solution $D$, $E$, $\langle W|$, $|V\rangle$
which satisfies the relations of Theorem \ref{ansatz}.
\end{theorem}

The above solution was related to Askey-Wilson polynomials.  Using 
properties of the 
Askey-Wilson integral, the authors calculated the 
\emph{current} $J_L$ of the finite ASEP.
Recall that  
$J_L = \frac{Z_{L-1}}{Z_L}$, where 
$Z_L = \langle W| (D+E)^L|V\rangle$.
Let $J = \lim_{L \to \infty} J_L$.

\begin{proposition}\cite[(6.5), (6.8), (6.11)]{USW}\label{prop:current}
Suppose that $q \neq 1$.  Let \begin{align*}
a&=\frac{1-q-\alpha+\gamma+\sqrt{(1-q-\alpha+\gamma)^2+4\alpha\gamma}}{2\alpha}
\text{ and }\\
b&=\frac{1-q-\beta+\delta+\sqrt{(1-q-\beta+\delta)^2+4\beta\delta}}{2\beta}.
\end{align*}
\begin{enumerate}
\item If $a>1$ and $a>b$ then 
  $J 
  =(1-q) \frac{a}{(1+a)^2}$.
\item If $a<1$ and $b<1$ then 
   $J = \frac{1-q}{4}$.
\item If $b>1$ and $b>a$ then 
  $J = (1-q) \frac{b}{(1+b)^2}.$
\end{enumerate}
\end{proposition}

\section{Formal definition of the semi-infinite ASEP}\label{semi-infinite}

We now define the semi-infinite ASEP.  Since this is a Markov process
with infinitely many states, one must define it carefully;  we 
give its Markov generator below.  
This Markov generator then determines a Markov semigroup
and hence a Markov process, see \cite[Chapter 1]{Liggett} or 
\cite[Section A.1]{G} for details.

Let $\eta=(\eta_1,\eta_2,\dots)$ be a state in $X$.
If $i$ is a positive 
integer, we define from $\eta$ two 
new states $\eta^i$ and $\eta^{i,i+1}$ by 
\begin{equation*}
(\eta^i)_j = 
 \begin{cases} 1-\eta_i & \text{ if $j=i$}\\
         \eta_j &\text{ if $j\neq i$}
\end{cases} 
\end{equation*}
\begin{equation*}
\text{ and }(\eta^{i,i+1})_j =
\begin{cases}
\eta_{i+1} & \text{ if $j=i$}\\
\eta_i    & \text{ if $j=i+1$}\\
\eta_j & \text{ if $j\neq i,i+1$}
\end{cases}
\end{equation*}

Let $C_0(X)$ be the set of {\it cylinder functions} on $X$,
i.e. functions from $X$ to $\R$ which depend on only finitely many 
sites.

\begin{definition}\label{L}
The Markov generator $\LL$ of the  semi-infinite ASEP 
is defined as follows.
Given any function $f\in C_0(X)$, 
\begin{align*}
\LL f(\eta) &= \alpha (1-\eta_1)(f(\eta^1)-f(\eta))+\gamma \eta_1 (f(\eta^1)-f(\eta))\\
&+\sum_{x=1}^\infty \left( \eta_x (1-\eta_{x+1}) (f(\eta^{x,x+1})-f(\eta))+q(1-\eta_x) \eta_{x+1}(f(\eta^{x,x+1})-f(\eta))\right).
\end{align*}
\end{definition}

We are interested in stationary measures of the corresponding
 Markov process.  A measure
$\mu$ is stationary if 
$\E^{\mu}(\LL f) = 0$
for all $f\in C_0(X)$.  
Here $\E^{\mu}$ is the expected value 
with respect to a measure $\mu$.  
Note that since the state space $X$ is infinite, the uniqueness
of the stationary measure is no longer assured.


\section{The matrix ansatz for the semi-infinite ASEP}\label{semiansatz}

We first prove a {\it matrix ansatz} in
the spirit of \cite{Derrida1}.
The version which we shall state and prove for the semi-infinite
ASEP is a  generalization
of a theorem of Grosskinsky \cite[Theorem 3.2]{G}; his ansatz is the same
as ours, except he set $\gamma=0$ and $q=0$. 

In what follows, we use the terminology \emph{signed measure} for a measure
which is not necessarily positive.  We will first 
give a matrix ansatz which describes stationary signed measures
(Theorem \ref{semi-ansatz}), and then in the following section,
we'll give a theorem 
(Theorem \ref{positivity}) which 
provides conditions guaranteeing that the 
measures are positive.

\begin{theorem}\label{semi-ansatz}
Suppose there are matrices $D,E$ and vectors $\langle W|, 
|V\rangle$, which satisfy
\begin{enumerate}
\item[(a.)] $DE-qED=c(D+E)$
\item[(b.)] $\alpha \langle W|E-\gamma \langle W|D = c\langle W|$
\item[(c.)] $(D+E)|V\rangle = |V\rangle.$
\end{enumerate}
Let $\eta = (\eta_1, \eta_2, \dots, \eta_L) \in \{0,1\}^L$.
Then the signed measure $\mu^L$ defined by 
\begin{equation} \label{signed-measure1}
\mu^L(\eta_1,\dots,\eta_L) = 
\frac{\langle W| \prod_{x=1}^L \eta_x D+(1-\eta_x) E|V\rangle}
{\langle W| (D+E)^L | V\rangle}
\end{equation}
is stationary for the process defined by  $\LL$.
Here the parameter $c$ determines the stationary current, i.e.
$\E^{\mu}( \eta_x (1 - \eta_{x+1})-q(1-\eta_x)\eta_{x+1}) = c$ for all $x\in \Z^+$.
\end{theorem}

\begin{remark}
The measure $\mu^L$ defined above does not depend 
on the choice of solution $D$, $E$, $\langle W|$, $|V\rangle$.
Indeed, for any word $Y$ in $D$ and $E$, 
by repeatedly applying relations (a.), (b.) and (c.),
one can express $\langle W|Y|V\rangle$ in terms of 
$\alpha$, $\gamma$, $q$, $c$, and $\langle W|V\rangle$.
\end{remark}




\begin{proof}
Suppose that $f\in C_0(X)$ concentrates on sites
$\{1,2,\dots,L\}$.  
Using Definition \ref{L}, the stationary condition which we must check
becomes:
\begin{align*}
0 &=\sum_{\eta} \alpha \mu^L(\eta)(1-\eta_1)(f(\eta^1)-f(\eta))+
    \sum_{\eta} \gamma \mu^L(\eta) \eta_1 (f(\eta^1)-f(\eta)) +\\
 &\hspace{.5cm}\sum_{\eta} \sum_{x=1}^{L-1} \left[ \mu^L(\eta) \eta_x (1-\eta_{x+1})(f(\eta^{x,x+1})-f(\eta)) +  q\mu^L(\eta)(1-\eta_x)\eta_{x+1}(f(\eta^{x,x+1})-f(\eta))\right] +\\
&\hspace{.5cm}\sum_\eta \left[\mu^{L+1}(\eta)\eta_L(1-\eta_{L+1})(f(\eta^{L,L+1})-f(\eta)) +
q \mu^{L+1}
(\eta)(1-\eta_L)\eta_{L+1}(f(\eta^{L,L+1})-f(\eta))\right].
\end{align*}
Here the sum is over all $\eta \in \{0,1\}^L$.
Rewriting this equation gives
\begin{align*}
0 &=\sum_{\eta} f(\eta) \bigg( \alpha \mu^L(\eta^1)(1-\eta^1_1)-
  \alpha \mu^L(\eta) (1-\eta_1)+\gamma \mu^L(\eta^1) \eta^1_1 
 -\gamma \mu^L(\eta) \eta_1) + \\
 &\hspace{.5cm}\sum_{x=1}^{L-1} [\mu^L(\eta^{x,x+1})\eta_x^{x,x+1} (1-\eta_{x+1}^{x,x+1})-\mu^L(\eta) \eta_x (1-\eta_{x+1})+q \mu^L(\eta^{x,x+1})(1-\eta_x^{x,x+1})
  \eta_{x+1}^{x,x+1}\\
 &\hspace{1cm}-q\mu^L(\eta)(1-\eta_x)\eta_{x+1} ]\\
 &\hspace{.5cm}+\mu^{L+1}(\eta^{L,L+1})\eta_L^{L,L+1} (1-\eta_{L+1}^{L,L+1} 
  -\mu^{L+1}(\eta) \eta_L (1-\eta_{L+1})\\
 &\hspace{.5cm}+q\mu^{L+1}(\eta^{L,L+1})(1-\eta_L^{L,L+1})\eta_{L+1}^{L,L+1}-q\mu^{L+1}(\eta)(1-\eta_L)\eta_{L+1} \bigg).
\end{align*}

Note that $\eta^1_1 = 1-\eta_1$.
The coefficient of $f(\eta)$ in the above equation is 
\begin{align*}
&\sum_{x=1}^{L-1} [\mu^L(\eta^{x,x+1}) \eta_{x+1}(1-\eta_x) 
  -\mu^L(\eta) \eta_x (1-\eta_{x+1})+q\mu^L(\eta^{x,x+1})(1-\eta_{x+1})\eta_x\\
&\hspace{.5cm} - q\mu^L(\eta)(1-\eta_x) \eta_{x+1}]\\
&+\mu^{L+1}(\eta^{L,L+1})\eta_{L+1}(1-\eta_L)-\mu^{L+1}(\eta) \eta_L(1-\eta_{L+1})\\
&+q\mu^{L+1}(\eta^{L,L+1})(1-\eta_{L+1})\eta_L - q\mu^{L+1}(\eta)(1-\eta_L) \eta_{L+1} \\
&+\alpha \mu^L(\eta^1)\eta_1 - \alpha\mu^L(\eta) (1-\eta_1)+\gamma \mu^L(\eta^1)(1-\eta_1)-\gamma \mu^L(\eta) \eta_1.
\end{align*}
We aim to show that each coefficient is equal to $0$.
Rearranging terms gives
\begin{align}\label{analyze}
&\sum_{x=1}^{L-1} \bigg[\mu^L(\eta^{x,x+1}) \eta_{x+1}(1-\eta_x) 
- q\mu^L(\eta)(1-\eta_x) \eta_{x+1}\\
  &\hspace{.7cm}-(\mu^L(\eta) \eta_x (1-\eta_{x+1})- \label{3}
    q\mu^L(\eta^{x,x+1})(1-\eta_{x+1})\eta_x)\bigg]\\
&+\mu^{L+1}(\eta^{L,L+1})\eta_{L+1}(1-\eta_L) \label{4}
- q\mu^{L+1}(\eta)(1-\eta_L) \eta_{L+1} \\
&-[\mu^{L+1}(\eta) \eta_L(1-\eta_{L+1}) \label{5}
-q\mu^{L+1}(\eta^{L,L+1})(1-\eta_{L+1})\eta_L]\\ 
&+[\alpha \mu^L(\eta^1)\eta_1  \label{6}
-\gamma \mu^L(\eta) \eta_1]
- [\alpha\mu^L(\eta) (1-\eta_1)-\gamma \mu^L(\eta^1)(1-\eta_1)].
\end{align}

Note that each configuration of particles can be seen as a 
sequence of empty and occupied blocks.  
Suppose that the first $L$ sites of $\eta$ consists of $n$ such 
blocks
$(\circ \dots \circ) (\bullet \dots \bullet)(\circ\dots \circ) \dots (\bullet \dots \bullet)(\circ\dots \circ)$
where there are $\tau_1$ $\circ$'s in the first block,
$\tau_2$ $\bullet$'s in the second block, \dots, and $\tau_n$
$\circ$'s in the last block. Here we assume that 
all $\tau_i$'s are nonzero, so in particular, 
the first $L$ sites of $\eta$ begin and end with $\circ$.
Thinking of the configuration of particles as a sequence of 
empty and occupied blocks, we also use $\tau$ to denote $\eta$.

At a boundary between a full and empty block ($\tau_i$ and 
$\tau_{i+1}$) we can apply the bulk rule of the ansatz to get 
$\tau-q\tau' = c(\tau^i + \tau^{i+1})$.  
Here, $\tau'$ is the configuration obtained from $\tau$ by 
swapping the adjacent $\bullet$ and $\circ$ in the $i$th and $i+1$st block,
and $\tau^i$ is obtained from $\tau$ by deleting one site in block $i$.

Noting that it has non-zero values only at the block boundaries, the 
sum over $x$ in \eqref{analyze} and \eqref{3} telescopes:
\begin{equation*}
\sum_{i=1, i \text{ odd}}^{n-2}
  c [\mu^{L-1}(\tau^i)+\mu^{L-1}(\tau^{i+1}) -(\mu^{L-1}(\tau^{i+1})+\mu^{L-1}(\tau^{i+2}))]
  = c\mu^{L-1}(\tau^1)-c\mu^{L-1}(\tau^n). 
\end{equation*}

Since we have assumed that the first $L$ sites of $\eta$ begin and end
with a $\circ$, we have that 
$\eta_1 = \eta_L=0$.  
Applying this and the relations of the ansatz  allows us to simplify 
the quantities \eqref{4}, \eqref{5} and \eqref{6}:
\begin{align*}
+\mu^{L+1}(\eta^{L,L+1})\eta_{L+1}(1-\eta_L)-q\mu^{L+1}(\eta)(1-\eta_L)\eta_{L+1} &= c(\mu^L(\eta)+\mu^{L}(\eta^L)),\\
-[\mu^{L+1}(\eta)\eta_L(1-\eta_{L+1})-q\mu^{L+1}(\eta^{L,L+1})(1-\eta_{L+1})\eta_L] &= 0,\\
+\alpha \mu^L(\eta^1) \eta_1 - \gamma \mu^L(\eta) \eta_1 &= 0,\\
-[\alpha \mu^L(\eta)(1-\eta_1) - \gamma \mu^L(\eta^1)(1-\eta_1)] &= 
  -\alpha \mu^L(\eta) + \gamma \mu^{L}(\eta^1).
\end{align*}
Therefore the coefficient of $f(\eta)$, which is given 
by \eqref{analyze} through \eqref{6}, simplifies to 
$$c\mu^{L-1}(\tau^1)-c\mu^{L-1}(\tau^n) + c\mu^L(\eta) + c\mu^L(\eta^L) 
  -\alpha \mu^L(\eta)+\gamma \mu^L(\eta^1).$$
But now note that 
by relation (c.) of the ansatz,
$ c\mu^L(\eta) + c\mu^L(\eta^L) =c\mu^{L-1}(\tau^n)$,
and by relation (b.) of the ansatz, 
$-\alpha \mu^L(\eta)+\gamma \mu^L(\eta^1) = -c\mu^{L-1}(\tau^1).$
It follows that the coefficient of $f(\eta)$ is $0$.

This completes the proof, when the first $L$ sites
of $\eta$ begin and end with $\circ$.  The proof is analogous
if the first $L$ sites begin or end with $\bullet$.
\end{proof}

\begin{remark}\label{semi-alternate}
In fact the above argument proves the following statement.
Suppose that $g:\{D,E\}^* \to \R$ is a function on 
words in $D$ and $E$ (extended linearly to linear combinations
of such words) such that for any words $Y$ and $Y'$ in 
$D$ and $E$, we have the following:
\begin{enumerate}
\item[(a.)] $g(Y(DE-qED)Y')=cg(Y(D+E)Y')$
\item[(b.)] $g(\alpha EY-\gamma DY) = cg(Y)$
\item[(c.)] $g(Y(D+E)) = g(Y).$
\end{enumerate}
Let $\eta = (\eta_1, \eta_2, \dots, \eta_L) \in \{0,1\}^L$.
Then the signed measure $\mu^L$ defined by 
\begin{equation}\label{signed-measure}
\mu^L(\eta_1,\dots,\eta_L) = 
\frac{g(\prod_{x=1}^L \eta_x D+(1-\eta_x) E)}
{g((D+E)^L)}
\end{equation}
is stationary for the process defined by  $\LL$.
Here the parameter $c$ determines the stationary current, i.e.
$\E^{\mu}( \eta_x (1 - \eta_{x+1})-q(1-\eta_x)\eta_{x+1}) = c$ for all $x\in \Z^+$.
\end{remark}

\section{Positivity of the measures} \label{sec:positivity}

One would like to know when the signed measure defined in 
\eqref{signed-measure1} or 
\eqref{signed-measure}
is positive.
\begin{theorem}\label{positivity}  
The signed measure defined in \eqref{signed-measure1} (equivalently, 
\eqref{signed-measure}) is positive provided that $q \leq 1$ and
one of the inequalities below is satisfied:
\begin{enumerate}
\item $a\geq 1$ and $c \leq (1-q)a/(1+a)^2$, or 
\item $a\leq 1$ and $c \leq (1-q)/4$.
\end{enumerate}
Here $a$ is defined as in Proposition \ref{prop:current}.
\vspace{-.1cm}
\end{theorem}

We will prove Theorem \ref{positivity} by finding a solution to 
the semi-infinite matrix ansatz (the version in 
Remark \ref{semi-alternate}) which is obtained as a limit
of a solution to the finite matrix ansatz.

\begin{proposition}\label{prop:limit}
Let $D$, $E$, $\langle W|$, $|V\rangle$ denote the 
solution to the finite matrix ansatz from Theorem \ref{USWsolution}.
Let $C = D+E$.
Then for any word $Y$ in $D$ and $E$, the following limit exists:
$$\lim_{m \to \infty} \frac{\langle W| YC^m |V\rangle}{\langle W|C^m|V\rangle}.$$
\end{proposition}

\begin{proof}
We will use relations (A.) and (B.) of the finite matrix ansatz
together with the fact (Proposition \ref{prop:current})
that $\lim_{m \to \infty} \frac{\langle W| C^{m-1} |V\rangle}{\langle W|C^m|V\rangle}$
exists.  Note that the latter fact implies that for any finite $\ell$,
both 
$\lim_{m \to \infty} \frac{\langle W| C^{\ell} C^{m-1} |V\rangle}{\langle W|C^m|V\rangle}$
and $\lim_{m \to \infty} \frac{\langle W|  C^{m-1} |V\rangle}{\langle W|C^{\ell} 
C^m|V\rangle}$ exist.

We will prove the result by induction on the length of $Y$.  
We start by considering the length $1$ case, i.e. $Y=D$ or $Y=E$.
Let $x_m = \frac{\langle W|EC^m |V \rangle}{\langle W|C^{m+1}|V\rangle}$
and $y_m = \frac{\langle W|DC^m |V \rangle}{\langle W|C^{m+1}|V\rangle}$.
Then we have $x_m + y_m=1$.  But also, by relation (B.) of the ansatz, we have
$\alpha x_m - \gamma y_m = \frac{Z_m}{Z_{m+1}}=J_{m+1}$.
We can therefore solve for $x_m$ and $y_m$ in terms of $J_{m+1}$;
since the limit of $J_{m+1}$ exists as $m \to \infty$, so does
the limit of $x_m$ and $y_m$.
It follows that for $Y = D$ or $Y=E$, the limit
$\lim_{m \to \infty} \frac{\langle W| YC^m |V\rangle}{\langle W|C^m|V\rangle}$
exists.

More generally, for any word $Y'$ of length $\ell(Y')>1$,
we will show that we can solve for 
$\frac{\langle W|Y' C^m |V\rangle}{\langle W|C^m |V\rangle}$ 
in terms of quantities of the form 
$\frac{\langle W|Y C^m |V\rangle}{\langle W|C^m |V\rangle}$ 
where the length $\ell(Y)$ of $Y$ is at most $\ell(Y')-1$.
This will complete the proof, since by the inductive hypothesis,
we can write the latter quantities 
in terms 
of the parameters $\alpha$, $\beta$, $\gamma$, $\delta$, $q$ and 
$J_m$'s, and hence can take the limit as $m$ goes to infinity.

Note that any word $Y'$ of length greater than $1$
can be written in the form $DY$ or $EY$ where 
the length of $Y$ is non-negative.
Using relation (B.) of the finite matrix ansatz, 
for any word $Y$ in $D$ and $E$, we have that 
\begin{equation}\label{eq:1}
\alpha \langle W| EYC^m | V\rangle - \gamma \langle W|DYC^m |V\rangle 
= \langle W|YC^m |V \rangle.
\end{equation}
And by repeatedly using relation (A.) of the ansatz, 
we can write
\begin{align*}
q^{\ell(Y)} \langle W|EYC^m|V\rangle &=
  q^{\# E's \text{ in }Y} \langle W|YEC^m|V\rangle +
\text{ terms of shorter length.}\\
\langle W|DYC^m|V \rangle &=
  q^{\# E's \text{ in }Y} \langle W|YDC^m|V\rangle +
\text{ terms of shorter length.}
\end{align*}
Here a {\it term of shorter length} means 
a monomial in the parameters times
a term of the form $\langle W|Y''C^m|V \rangle$
where $\ell(Y'') < \ell(YE) = \ell(YD).$

Summing the last two equations gives
\begin{equation}\label{eq:2}
q^{\ell(Y)} \langle W|EYC^m|V\rangle +
\langle W|DYC^m|V \rangle =
  q^{\# E's \text{ in }Y} \langle W|YC^{m+1}|V\rangle +
\text{ terms of shorter length.}
\end{equation}
But now since the right-hand sides of equations
\eqref{eq:1} and \eqref{eq:2} are known quantities,
and the determinant of the coefficient matrix is 
$\alpha + \gamma q^{\ell(Y)}$ which is nonzero,
we can solve for 
$ \langle W|EYC^m|V\rangle $ and 
$\langle W|DYC^m|V \rangle$.  This completes the proof. 
\end{proof}

\begin{proposition}\label{prop1}
Suppose that $q\neq 1$.  
Let $D$, $E$, $\langle W|$, and $|V\rangle$ be as in 
Theorem \ref{USWsolution}, and set $C = D+E$.
Let $c=J$ (recall that $J$ is given by Proposition
\ref{prop:current}, depending on three cases).
Denote the length of $Y$ by $\ell(Y)$.  
For each word $Y$ in $D$ and $E$, define 
$$g(Y) = c^{\ell(Y)} \lim_{m \to \infty}
\frac{\langle W|YC^m |V\rangle}{\langle W| C^m |V\rangle}.$$
Then $g(Y)$ satisfies the relations of 
Remark \ref{semi-alternate}.
\end{proposition}
\begin{proof}
By Proposition \ref{prop:limit}, the definition of $g(Y)$ makes 
sense.  Now note that 
the relations (a.) and (b.) of Remark \ref{semi-alternate}
follow directly from relations (A.) and (B.) of Theorem \ref{ansatz}.
To check relation (c.), note that 
\begin{align*}
g(Y(D+E)) &= c^{\ell(Y)+1} \lim_{m\to \infty} \frac{\langle W|Y(D+E)C^m|V\rangle}{\langle W|C^m |V \rangle} \\
    &= c^{\ell(Y)+1} \lim_{m\to \infty} \frac{\langle W|YC^{m+1}|V\rangle}{\langle W|C^m |V \rangle} \\
    &= c^{\ell(Y)+1} \lim_{m\to \infty} \frac{\langle W|YC^{m+1}|V\rangle}{\langle W|C^{m+1} |V \rangle}  \cdot 
\frac{\langle W|C^{m+1}|V\rangle}{\langle W|C^{m} |V \rangle} \\
    &= c^{\ell(Y)} \lim_{m\to \infty} \frac{\langle W|YC^{m+1}|V\rangle}{\langle W|C^{m+1} |V \rangle} \\
   &=g(Y).
\end{align*}
\end{proof}

Finally we turn to the proof of Theorem \ref{positivity}.
\begin{proof}
First note that if $q=1$ (hence $c=0$), it is easy to check that 
the measure defined by \eqref{signed-measure} is positive.  This 
can be checked directly from the ansatz relations, which become
very simple.

We now consider the case that $q<1$.
By Proposition \ref{prop1}, we can define a function $g$
which satisfies the relations of Remark \ref{semi-alternate}.  Therefore
the signed measure defined by \eqref{signed-measure}
is stationary for the process defined by $\LL$.
Moreover, 
$$\frac{\langle W|YC^m |V \rangle}{\langle W|C^{\ell(Y)+m}|V \rangle}$$
is non-negative, because it is a correlation function in the finite ASEP.
Since $c$ is positive, and $g(Y)$ is a limit of non-negative values,
it follows that 
$g(Y)$ is non-negative.
Therefore the signed measure defined by \eqref{signed-measure}
using the function $g$ from Proposition \ref{prop1} is a positive measure.

We now need to check that $c$ satisfies one of the 
inequalities in Theorem \ref{positivity}, and that indeed,
any pair of $a$ and $c$ satisfying these inequalities can be obtained
from Proposition \ref{prop1} using a suitable choice of $\alpha$, $\beta$, 
$\gamma$, and $\delta$.

Recall that $J$ was computed in Proposition \ref{prop:current},
and that we have set $c=J$.  Therefore
in the first two cases, $c$ satisfies the inequalities 
of Theorem \ref{positivity}.
We now consider the third case.  
If $a \geq 1$ and $b>a$ then it follows that 
$$c =J = \frac{(1-q)b}{(1+b)^2} \leq \frac{(1-q)a}{(1+a)^2}.$$
While if $a \leq 1$ and $b>a$ then 
$$c= J = \frac{(1-q)b}{(1+b)^2} \leq \frac{1-q}{4}.$$
In both cases, it follows that $c=J$ satisfies the 
inequalities of Theorem \ref{positivity}.  

Moreoever, if we let $b$ tend to infinity 
(one may achieve this by sending $\beta$ to $0$), then
$c = \frac{(1-q)b}{(1+b)^2}$ tends to $0$.
Therefore it is possible to  choose appropriate
values $\alpha$, $\beta$, $\gamma$, $\delta$ and $q$
so as to realize any pair $(a,c)$ satisfying 
the conditions of Theorem \ref{positivity}.
By the previous arguments, the corresponding measure 
that we get in this case is positive.
\end{proof}

\section{Proof of Theorem \ref{mainthm}}\label{proof}

In this section we prove Theorem \ref{mainthm}.

Note that if we set $\beta=c$ and $\delta=-c$ in Theorem 
\ref{ansatz}, then the relations
(A.), (B.), and (C.) become:
\begin{enumerate}
\item[(A'.)] $DE-qED=D+E$
\item[(B'.)] $\alpha \langle W|E-\gamma \langle W|D = \langle W|$
\item[(C'.)] $ (D+E)|V\rangle  = \frac{1}{c} |V\rangle.$
\end{enumerate}

Note that these relations are nearly identical to the relations
(a.), (b.) and (c.) from the Matrix Ansatz for the semi-infinite ASEP.

\begin{proposition}\label{measures}
Suppose that $d, e$, $\langle w|$ and $|v \rangle$ satisfy
(a.), (b.), and (c.), and suppose that
$D, E$, $\langle W|$, and $|V\rangle$ satisfy
(A'.), (B'.), and (C'.).
Let $y$ be an arbitrary word of length $\ell$ in $d$ and $e$,
and let $Y$ be the corresponding word in $D$ and $E$.
Then $$\frac{\langle w|y|v\rangle}{\langle w|(d+e)^{\ell}|v\rangle}
 =  \frac{\langle W|Y|V\rangle}{\langle W|(D+E)^{\ell}|V\rangle}.$$ 
\end{proposition}

\begin{proof}
Let $\widetilde{D} = cD$ and $\widetilde{E} = cE$.  Since
$D, E$, $\langle W|$, and $V \rangle$ satisfy 
(A'.), (B'.), and (C'.), it is easy to verify that
$\widetilde{D}, \widetilde{E}$, $\langle W|$, and $V \rangle$ satisfy 
(a.), (b.), and (c.).  
We also have that $d, e$, $\langle w|$, and $|v \rangle$ satisfy 
(a.), (b.), and (c.).  
Therefore both of them yield the same measure, as defined in 
Theorem \ref{semi-ansatz}.
Letting $\widetilde{Y}$ denote the word in 
$\widetilde{D}$ and $\widetilde{E}$ corresponding to $Y$, 
we have that 
$$\frac{\langle w|y|v\rangle}{\langle w|(d+e)^{\ell}|v\rangle}
 =  \frac{\langle W|\widetilde{Y}|V\rangle}{\langle W|(\widetilde{D}+\widetilde{E})^{\ell}|V\rangle}
 =  \frac{\langle W|Y|V\rangle}{\langle W|(D+E)^{\ell}|V\rangle}.$$ 
\end{proof}

Theorem \ref{mainthm} now follows from 
Proposition \ref{measures}, and Theorems \ref{semi-ansatz} and \ref{ansatz}.

\section{Conclusion}\label{conclusion}

In this paper we have given a combinatorial interpretation
for the stationary measures of the semi-infinite ASEP.
More specifically, one may compute the finite correlation functions
of the stationary measures using sums over 
\emph{staircase tableaux}, with the parameters
$\alpha$, $\beta=c$, $\gamma$, $\delta=-c$, and $q$.

In particular, 
we have demonstrated that a rather nonsensical specialization of 
the stationary distribution 
of the finite ASEP  -- the specialization $\delta = -\beta$ -- can 
be given a meaningful interpretation in terms of the 
ASEP on a semi-infinite lattice.

One might ask more generally when this phenomenon can occur. For concreteness,
in the discussion below, we will consider \emph{finite} Markov chains.
\begin{itemize}
\item  Consider a Markov chain $M$ whose transition matrix is written in terms of 
  one or more parameters (e.g. hopping rates).  
  Typically we don't consider $M$ to ``make sense"
  unless these parameters are non-negative.
\item  Recall that the stationary distribution $\mu$ of a Markov chain is the 
  unique left eigenvector of the transition matrix associated with eigenvalue $1$.
\item  One may choose a specialization of the parameters and consider the corresponding
  specialization of $\mu$.
\item  If one makes one or more parameters negative (or even complex),  when can
  one still give a probabilistic or physical meaning to the corresponding ``stationary distribution,"
  that is, the corresponding specialization of $\mu$?
\end{itemize}

\end{document}